\documentclass{article}
\usepackage{array, paralist, enumerate, amsmath, amsfonts, amssymb, amscd, color, mathrsfs,comment}
\usepackage{amsthm} 

\usepackage{geometry}
\usepackage{framed}
\usepackage{hyperref}
\usepackage{graphicx}
\usepackage{epstopdf}
\usepackage[all,cmtip]{xy}

\usepackage{setspace}
\usepackage{etoolbox}
\AtBeginEnvironment{quote}{\singlespace\vspace{-\topsep}\small}
\AtEndEnvironment{quote}{\vspace{-\topsep}\endsinglespace}

\numberwithin{equation}{section}
\newtheorem{thm}[equation]{Theorem}

\newtheorem{prop}[equation]{Proposition}

\theoremstyle{definition}

\newtheorem{remark}[equation]{Remark}
\newtheorem{protocol}[equation]{Protocol}

\usepackage{tcolorbox}

\title{Folding Custom Gates with Verifier Input}

\author{Aard Vark\footnote{\url{tideofwords@gmail.com}} \and Yan X Zhang\footnote{\url{yan.x.zhang@sjsu.edu}, Department of Mathematics and Statistics, San Jos\`{e} State University}}

\date{January 20, 2024}

\begin{document}

\maketitle

\begin{abstract}
    In the context of interactive proofs, a \emph{folding scheme}, popularized by Nova, is a way to combine multiple instances of a constraint system into a single instance, so the validity of the multiple instances can statistically be reduced to the validity of a single one. We show how Nova folding can be generalized to ``custom'' gates and extra rounds of verifier randomness. As an application of this extension, we present \emph{Origami}, the first (to our knowledge) known example of a folding scheme for lookups. 
\end{abstract}

\section{Introduction}

In recent years, the field of interactive proofs (IPs) and zero-knowledge (ZK) have been very active thanks to new applications arising from blockchain and cryptocurrencies. One of the more exciting new developments is the concept of \emph{folding}. A \emph{folding scheme} is a protocol to aggregate multiple satisfying instances of an \emph{arithmetization} (some constraint system, such as a set of equations) into a single instance. Usually, the constraint systems are encoding some computation done by a prover to be checked by a verifier. Therefore, a folding scheme then allows  multiple computations to be combined and just verified once, which can be a significant saving for many implementations of proving systems.

Nova \cite{nova} introduced folding and a folding scheme for R1CS circuits (a particular arithmetization). Under the hood, Nova-style folding works by taking a random linear combination of the instances; the R1CS equations need to be generalized to \emph{relaxed R1CS} constraints, which are stable under linear combination.

Our goals in this work are:
\begin{enumerate}
    \item Explicitly show how Nova-style folding can be generalized to any polynomial ``custom gate''. The idea is already sketched in Sangria \cite[Section 3.3]{sangria}; we are just filling in details.
    \item Generalize how the protocol works to when we have extra rounds of verifier randomness in the protocol. We call this \emph{custom gates with verifier input.}
    \item Show that folding Halo2-style lookups can be seen as a special case of custom gates with verifier input. We explicitly write out the protocol (\textbf{Origami}). We did not observe any previous examples of folding schemes for lookups in literature.
\end{enumerate}

This paper is an expanded version of a note posted on HackMD earlier in April 2023 \cite{hackmd}. Similar ideas were developed independently in Protostar \cite{protostar}, under the name of ``accumulation schemes for special-sound protocols.''

\section{Preliminaries}

Generally, we use capital letters as vectors, and the matching lower case letters refer to their coordinates. For example, a vector $S$ would have its entries labeled $(s_1, \ldots, s_n)$.

\subsection{Arithmetization and PAIRs}
\label{sec:PAIR}

In interactive proofs, an \emph{arithmetization} is an embedding of a computation (such as $3 + 4 = 7$) into some algebraic constraint system (such as a polynomial or a linear equation). A \emph{preprocessed algebraic intermediate representation (PAIR)} (as in \cite{AIR-to-RAPs}) is a particular arithmetization of the following form:
\begin{enumerate}
    \item $P$ (the \emph{structure}) is a collection of polynomials $f_1, \ldots, f_\ell$ in $2w$ variables  over a finite field $\mathbb{F}$.
    \item   An \emph{execution trace} $T$ for $P$ is an array of elements $T_{r, c}$ of $\mathbb{F}$, with $n$ rows labeled $0$ to $n-1$, each of width $w$. 
    \item In order for the trace $T$ to be valid, it is required that, when we substitute into any $f_i$ the $2w$ elements of any two consecutive rows of $T$, the result equals zero:
$$ f_i(T_{j, 1}, T_{j, 2}, \ldots, T_{j, w}, T_{j+1, 1}, T_{j+1, 2}, \ldots, T_{j+1, w}) = 0, $$ 
where $j+1$ wraps to $0$ if $j = n-1$.
    \item The first $t \leq w$ of the $w$ columns of $T$ will be predefined and publicly known.  These $k$ columns are known as the PAIR \emph{instance}.
    \item The remaining $w-k$ columns of $T$ are known as the \emph{witness}.
\end{enumerate}
    
From a theoretical viewpoint, the requirement of two consecutive rows of the last item is unnecessarily restrictive; the main idea is that there is some set of elements $S_j$ ``corresponding'' to each row and that $S_{j+1}$ is a vertical shift of $S_j$. In practice, $T$ is a recording of computation, where each two consecutive rows are embedding one particular computation. 

As an example of a PAIR (slightly modified) from \cite{AIR-to-RAPs}, see Table~\ref{tab:example_pair}. In this example, we have $n=4, t=1, w=2$. $P$ has a single polynomial $f_1,$ which we can write in shorthand as
\[
f_1 = C_1 (X_1[1] - (X_1 + X_2)) + (1-C_1) (X_1[1] - X_1*X_2).
\]
Explicitly, this means for all $j \in \{0, \ldots, 3\}$, we have
\[
f_{1, j} = c_{1, j} (x_{1, j+1} - (x_{1, j} + x_{2, j}) + (1 - c_{1, j}) (x_{1, j+1} - x_{1, j} * x_{2, j}) = 0,
\]
where:
\begin{itemize}
    \item the columns correspond to the vectors $C_1, X_1, X_2$ respectively; 
    \item recall that the lower case variables are the parts of the capital letter vectors. For example, $X_1 = (x_{1, 1}, \ldots, x_{1, 4})$;
    \item $X_1[1]$ means that the variable corresponds to the $X_1$ variable for the next row. 
\end{itemize}
Here, the first ($t=1$) preprocessed column is acting as a ``selector column,'' meaning it selects whether the operation is addition (corresponding to $c_{1, j} = 1$) or multiplication (corresponding to $c_{1, j} = 0$). The information of this column and $P = \{f_1\}$ together encode a circuit that is computing $(1 + 1 + 5)*3 = 21$ step by step, starting with adding the first $2$ inputs ($x_{1,1} = 1$ and $x_{2,1} = 1$) to obtain the first intermediate step $2$. The ``witness to the computation'' is then the content of the second through fourth ($w=3$) columns, and the combined table is the execution trace of the computation.

\begin{table}[]
    \centering
    \begin{tabular}{|c|c|c|}
        \hline
        1 & 1 & 1  \\
        \hline
        1 & 2 & 5 \\
        \hline
        0 & 7 & 3 \\
        \hline
        0 & 21 & 0 \\
        \hline
    \end{tabular}
    \caption{An example of a PAIR that computes $(1 + 1 + 5)*3 = 21$.}
    \label{tab:example_pair}
\end{table}

\subsection{Commitments}

We assume the basic definition of a \emph{commitment scheme} as in \cite{BCC} (where commitment schemes are referred to as ``blobs''). We assume that we have an \emph{additively homomorphic} commitment scheme $\mathrm{Com}(pp, x, \rho)$, meaning that $$\mathrm{Com}(pp, x, \rho) + \mathrm{Com}(pp, y, \rho) = \mathrm{Com}(pp, x+y, \rho)$$. We denote a commitment to a vector $V$ by $\overline{V} = \mathrm{Com}(pp, V, \rho)$. 

For clarity of notation, sometimes we will write $\mathrm{Com}(V) := \mathrm{Com}(pp, V, \rho)$ when the arguments are obvious from context. 

\subsection{Folding}

\emph{Folding schemes}, introduced in \cite{nova}, reduce the task of checking two instances in a relation $\mathcal{R}$ to the task of checking a single instance $\mathcal{R}$.  A formal definition is given in \cite[Section 3]{nova}.  

Given two different instance-witness pairs to the same structure, a folding scheme allows one to combine them into a single instance-witness pair.  This reduces the task of checking two traces to the task of checking a single trace.

In both \cite{nova} and the present work, folding is achieved by taking a random linear combination of the elements $T_{r, c}$ of the execution trace.
The constraint polynomials need to be \emph{relaxed} so as to be compatible with linear combinations.  A detailed construction is given in Section \ref{sec:relaxed}.

\subsection{Homogenization}

We start with a general algebraic identity.

\begin{prop}
\label{prop:delta}
Let $p(\mathbf{x})$ be a homogenous polynomial of degree $d$ in $n$ variables (where $\mathbf{x} = (x_1, x_2 \ldots, x_n)$). Then there are $d-1$ polynomials $\Delta^1 p, \Delta^2 p, \ldots, \Delta^{d-1} p$ of degree at most $d$ in $2n$ variables such that
$$ p(\mathbf{x} + r \mathbf{y}) = p(\mathbf{x}) + r^d p(\mathbf{y}) + \sum_{k=1}^{d-1} r^k \Delta^k p (\mathbf{x}, \mathbf{y}). $$
\end{prop}
\begin{proof}
Write $p(\mathbf{x})$ as a linear combination of monomials $x_{i_1} x_{i_2} \cdots x_{i_d}$:
$$p(\mathbf{x}) = \sum a_{i_1, i_2, \ldots, i_d} x_{i_1} x_{i_2} \cdots x_{i_d}.$$
 (It's OK if some of the indices are equal to each other.  For example, the monomial $x_1^2$ has $i_1 = i_2 = 1$. Also, this expression is not unique, and that's OK too.  For example, $x_1 x_2$ could be rewritten as $x_2 x_1$ or even $2x_1 x_2 - x_2 x_1$.)

The polynomial $\Delta^k p(\mathbf{x})$ can be computed as follows: for any monomial, $\Delta^k(x_{i_1} x_{i_2} \cdots x_{i_d})$ is the sum of the $\binom{d}{k}$ terms obtained by replacing $k$ of the $d$ variables $x_{i_j}$ with $y_{i_j}$.  For example:
\begin{align}
 \Delta^1 (x_1 x_2) & = x_1 y_2 + y_1 x_2 \\
 \Delta^2 (x_1 x_2 x_3) & = x_1 y_2 y_3 + y_1 x_2 y_3 + y_1 y_2 x_3 \\
 \Delta^1 (x_1^3) & = \Delta^1 (x_1 x_1 x_1) = x_1 y_1 y_1 + y_1 x_1 y_1 + y_1 y_1 x_1. \qedhere
\end{align}
\end{proof}

What if $p$ is not homogenous? To match the notation in Nova, we'll homogenize by introducing an extra variable: for any polynomial $p(\mathbf{x})$ of degree at most $d$ in $n$ variables, we'll define a homogenous polynomial $p^{homog}(\mathbf{x}, u)$ of degree $d$ in $n+1$ variables, as follows. Write $$p(\mathbf{x}) = \sum_{d' = 0}^{d} p_{d'}(\mathbf{x}), $$ where $p_{d'}(\mathbf{x})$ is homogenous of degree $d'$. In other words, $p_{d'}(\mathbf{x})$ is just the sum of all the degree-$d'$ terms in $p(\mathbf{x})$. Then define $$p^{homog}(\mathbf{x}, u) = \sum_{d' = 0}^{d} u^{d - d'} p_{d'}(\mathbf{x}).$$

In other words, to make $p^{homog}(\mathbf{x}, u)$ from $p(\mathbf{x})$, just multiply each term of $p(\mathbf{x})$ by whatever power of $u$ is needed to make the degree come out to $d$.

For future reference, it's worth noting that 
$$p(\mathbf{x}) = p^{homog}(\mathbf{x}, 1).$$

\section{The Protocol}

We describe a folding scheme for AIRs. The folding scheme we're about to describe works for PAIRs and more general polynomial constraint systems as well.  All that matters is that a valid trace be defined by some fixed collection of polynomial equations in the entries.  It is not necessary that each polynomial only involve entries from two consecutive rows.  Nonetheless, we'll continue working within the AIR framework.

To start, we introduce the main object, a \emph{relaxed AIR}, in the style of \cite{sangria}. The main idea of a relaxed AIR, like the relaxed R1CS analogue from Nova, is a polynomial constraint that involves a ``slack term.''

\subsection{Relaxed AIR Instances}
\label{sec:relaxed}

Recalling the setting of Section~\ref{sec:PAIR}, we write $f_i(T)$ for the length-$n$ vector whose $j$-th entry is $f_i$ applied to rows $j$ and $j+1$ of $T$:
$$ f_i(T_{j, 1}, T_{j, 2}, \ldots, T_{j, w}, T_{j+1, 1}, T_{j+1, 2}, \ldots, T_{j+1, w}). $$
(recall that rows wrap around, so row $n$ is the same as row $0$.)

We define a \emph{relaxed AIR instance} $$(P = \{f_i\}_{i \in \{1, \ldots, l\}}, \{E_i \in \mathbb{F}^n\}_{i \in \{1, \ldots, \ell}, u \in \mathbb{F})$$ to be satisfied by a trace $T$ (as before, a $n$-by-$w$ array of elements of $\mathbb{F}$) if
$$ f_i^{homog}(T_{j, 1}, T_{j, 2}, \ldots, T_{j, w}, T_{j+1, 1}, T_{j+1, 2}, \ldots, T_{j+1, w}, u) = (E_i)_j, $$ for each $i$ and $j$. We call the $E_i$ the \emph{slack vectors} and $T$ the \emph{witness} for the instance.

We'll write $f_i^{homog}(T, u)$ for the vector whose $j$-th entry is
$$ f_i^{homog}(T_{j, 1}, T_{j, 2}, \ldots, T_{j, w}, T_{j+1, 1}, T_{j+1, 2}, \ldots, T_{j+1, w}, u). $$

Since $f_i^{homog}$ is a polynomial in $2w+1$ variables, Proposition~\ref{prop:delta} defines polynomials $\Delta^k f_i^{homog}$ in $2(2w+1)$ variables.  We'll write $\Delta^k f_i^{homog} (T^1, u^1; T^2, u^2)$ for the vector whose $j$-th entry is
$$ \Delta^k f_i^{homog} (T^1_{j, 1}, \ldots, T^1_{j+1, w}, u^1, T^2_{j, 1}, \ldots, T^2_{j+1, w}, u^2). $$

Any AIR instance can be ``promoted'' to a relaxed AIR instance by setting $u = 1$ and $E = 0$. The main idea is that by converting AIR instances to relaxed AIR instances, the computations can be more easily combined or ``folded.''

\subsection{Committed relaxed AIR}

A \emph{committed relaxed AIR instance-witness pair} $(\overline{T}, u, \overline{E}, \rho)$ for a structure $P = \{f_1, \ldots, f_{\ell}\}$ consists of:
\begin{itemize}
\item a commitment $\overline{T}$ to an $n$-by-$w$ matrix of scalars,
    \item a scalar $u$,
    \item a commitment to \emph{slack vectors} $E = (E_1, \ldots, E_{\ell})$, and
\item randomness $\rho$ (to be used by the commitment scheme).
\end{itemize} 

A \emph{witness} to the committed relaxed AIR instance $(f_i, \overline{E_i}, \overline{T})$ is a tuple $(E, T)$ of 
\begin{itemize}
\item a slack vector $E$ and 
\item an $n$-by-$w$ matrix $T$,
\end{itemize}
such that $\overline{T} = \operatorname{Com}(pp, T, \rho)$, $\overline{E} = \operatorname{Com}(pp, E, \rho)$, and $f_i^{homog}(T, u) = E_i$.

\subsection{Single Fold for Custom Gates}
\label{sec:single-fold}

We'll describe a folding scheme for relaxed AIR instances.  Suppose a prover and a verifier have an AIR $P$, i.e.\ they both know the collection of polynomials $f_i$.  The prover is given two relaxed AIR instances $(T^1, E^1, u^1)$ and $(T^2, E^2, u^2)$. The verifier knows only the scalars $u^1$ and $u^2$.

Let $d_i$ be the degree of the polynomial $f_i$, and let $\Delta^k f_i^{homog}$ be as defined above.

The prover $\mathcal{P}$ and the verifier $\mathcal{V}$ carry out the following protocol.

\begin{tcolorbox}
\begin{protocol} [Custom Gate, Single Fold]
\label{protocol:single-fold}
\hfill
\medskip

INPUT [to $\mathcal{P}$]: 2 relaxed instance-witness pairs $I^{(i)} = (u^{(i)}, T^{(i)}, E^{(i)})$, $i \in \{1, 2\}$.

\medskip

OUTPUTS: [from  $\mathcal{P}$]: 1 relaxed instance-witness pair $I = (u, T, E)$; [from $\mathcal{V}$]: 1 relaxed committed instance $\overline{I} = (u, \overline{T}, \overline{E})$.

\begin{enumerate}
    \item $\mathcal{P}$ computes commitments $\overline{X}$ for $X \in \{T^{(1)}, T^{(2)}, E^{(1)}, E^{(2)}\}$ and their openings $\rho_X$ and sends them to $\mathcal{V}$. (explicitly: for all $X$, $\overline{X} = \operatorname{Com}(pp, X, \rho_X)$ holds.)
    
    \item For each constraint polynomial $f_i$, and each degree $1 \leq k \leq d_i - 1$, $\mathcal{P}$ also computes the \emph{cross term}
$$ B_{i, k} = \Delta^k f_i^{homog} (T^1, u^1; T^2, u^2) $$ (a vector of length $n$). $\mathcal{P}$ computes commitments and openings for $\overline{B_{i, k}}$ and sends them to $\mathcal{V}$.

    \item $\mathcal{V}$ samples a random \emph{folding challenge} $r \in \mathbb{F}$, and sends $r$ to $\mathcal{P}$.

    \item Both $\mathcal{P}$ and $\mathcal{V}$ compute $u = u^1 + r u^2.$ $\mathcal{P}$ computes $T = T^1 + r T^2$ and, for each $i$, $$E_i = E^1_i + r^{d_i} E^2_i + \sum_{k=1}^{d_i-1} r^k B_{i, k}.$$ $\mathcal{P}$ returns the ``folded'' relaxed AIR instance-witness $(u, T, E)$. 
    
    \item $\mathcal{V}$ computes $\overline{T} = \overline{T^1} + r \overline{T^2},$ and, for each $i$, $$\overline{E_i^1} + r^{d_i} \overline{E_i^2} + \sum_{k=1}^{d_i-1} r^k \overline{B_{i, k}}.$$ $\mathcal{V}$ returns the folded relaxed committed instance $(u, \overline{T}, \overline{E})$.
\end{enumerate}

\end{protocol}
\end{tcolorbox}

Like Nova, this folding process can be iterated.  Steps 2-5 allow two committed relaxed AIR instances to be folded into one, at the cost of sending just $4 + \sum_{i=1}^{\ell} d_i-1$ commitments to the verifier.  One can iterate the procedure to fold an arbitrary number of committed relaxed AIR instances together, one by one.  At the end, the prover only has to convince the verifier that the final $(E, T)$ is a witness to the folded instance. Explicitly:

\begin{tcolorbox}
    \begin{protocol}[Custom Gates, Full Protocol]
\hfill
\medskip

INPUT: [to $\mathcal{P}$]: $N$ instance-witness pairs / traces $\widetilde{I^{(i)}} = (T^{(i)})$, for $i \in \{1, \ldots, N\}$. 
\hfill
\medskip
OUTPUTS: [from $\mathcal{P}$]: 1 relaxed instance-witness pair $I^{cml} = (u, T, E)$;  [from $\mathcal{V}$]: 1 folded relaxed committed instance $\overline{I^{cml}} = (u, \overline{T}, \overline{E})$.

\begin{enumerate}
    \item To initialize folding, $\mathcal{P}$ initializes a ``cumulative lookup instance'' $I^{cml}$ where $u, T, E$ are all equal to zero. Similarly, $\mathcal{V}$ initializes a cumulative committed instance $\overline{I^{cml}}$ where $u,  \overline{T}, \overline{E}$ are all equal to zero.
    \item When we fold in a new lookup instance $\widetilde{I^{(i)}} = (T^{(i)})$, $\mathcal{P}$ constructs a relaxed lookup instance $I^{(i)} = (u^{(i)}, T^{(i)},  E^{(i)})$ by setting $u^{(i)} = 1$ and $E_j^{(i)} = (0, \ldots, 0)$ for all $j$.
    \item $\mathcal{P}$ and $\mathcal{V}$ run Protocol~\ref{protocol:single-fold}.
    \begin{itemize}
        \item $\mathcal{P}$ overwrites
$$ I^{cml} \leftarrow \texttt{SingleFold}_{\mathcal{P}}(I^{cml}, I^{(i)}),$$
that is: apply Protocol~\ref{protocol:single-fold}, using the current $I^{cml}$ as the first input and the new relaxed instance $I^{(i)}$ as the second input. During each step, we need a new $r = r^{(i)}$ for our ``folding randomness.'' 
\item $\mathcal{V}$ builds the committed relaxed instance
$\overline{I^{(i)}} = (u^{(i)},  \overline{T^{(i)}}, \overline{E^{(i)}})$
out of commitments sent from $\mathcal{P}$.  Using $\overline{I^{(i)}}$, $\mathcal{V}$ overwrites
$$ \overline{I^{cml}} \leftarrow \texttt{SingleFold}_{\mathcal{V}}(\overline{I^{cml}}, \overline{I^{(i)}}).$$
    \end{itemize}
    \item  After $N$ steps, the folding is complete.  All that remains is for $\mathcal{P}$ to convince $\mathcal{V}$ that the final folded tuple $I^{cml}$ is a legitimate witness to the committed instance $\overline{I^{cml}}$, which is done as in Halo2.
\end{enumerate}
    \end{protocol}
\end{tcolorbox}

\subsection{Custom Gates with Verifier Input}
\label{sec:cgvi}

The work in Section~\ref{sec:single-fold} is basically making explicit a sketch already outlined in Sangria \cite{sangria}. To achieve our main goal of allowing lookups, we need to generalize the context to one where the prover would need information from the verifier (usually in the form of random challenges) to compute certain parts of the trace in the middle of the computation (see RAPs \cite{AIR-to-RAPs}). More specifically, we assume that:
\begin{enumerate}
    \item The prover fills out the \emph{partial trace} $T_0$, a $n \times w_0$ table without additional information.
    \item The prover receives some information from the verifier $R$, which we assume is embedded as a vector in $F^{v}$ for some constant $v \in \mathbb{N}$. 
    \item With $S$, the prover is able to compute $w_1$ additional columns to create the \emph{full trace} $T$, a $n \times w$ table which can then be checked for validity against the structure $P$.
\end{enumerate}
We call this setup \emph{custom gates with verifier input}. The standard context where this workflow is used is lookups, our main goal. Thus, we first give the formalism and defer the example to Section~\ref{sec:lookups}.

\begin{tcolorbox}
    \begin{protocol}[Custom Gates with Verifier Input, Full Protocol]
\label{protocol:full-cgvi}
\hfill
\medskip
    
INPUT: [to $\mathcal{P}$]: $N$ partial instance-witness pairs / traces $\widetilde{I^{(i)}} = (T_0^{(i)})$, for $i \in \{1, \ldots, N\}$. 
\hfill
\medskip
OUTPUTS: [from $\mathcal{P}$]: 1 relaxed (full) instance-witness pair $I^{cml} = (T, u, E, R)$; [from 
 $\mathcal{V}$]: 1 folded relaxed committed instance $\overline{I^{cml}} = (\overline{T}, u, \overline{E}, R)$.

\begin{enumerate}
    \item To initialize folding, $\mathcal{P}$ initializes a ``cumulative lookup instance'' $I^{cml}$ where $T, u, E, R$ are all equal to zero. Similarly, $\mathcal{V}$ initializes a cumulative committed instance $\overline{I^{cml}}$ where $ \overline{T}, u, \overline{E}, R$ are all equal to zero.
    \item When we fold in a new partial trace $\widetilde{I^{(i)}} = (T_0^{(i)})$, 
    \begin{enumerate}
        \item $\mathcal{P}$ first commits to $\overline{T_0^{(i)}}$ for $i \in \{1,2\}$, sending the commitments to $\mathcal{V}$.
        \item $\mathcal{V}$ then gives $\mathcal{P}$ some information $R^{(i)}$.
        \item $\mathcal{P}$ uses $R^{(i)}$ to compute $T^{(i)}$, thus obtaining a full relaxed lookup instance $I^{(i)} = (u^{(i)}, R^{(i)}, T^{(i)},  E^{(i)})$ by setting $u^{(i)} = 1$, $E_j^{(i)} = (0, \ldots, 0)$ for all $j$, and $R^{(i)} = R$.
        \item $\mathcal{P}$ and $\mathcal{V}$ run Protocol~\ref{protocol:single-fold}.
    \end{enumerate}
    \item 
    \begin{itemize}
        \item $\mathcal{P}$ overwrites
$$ I^{cml} \leftarrow \texttt{SingleFold}_{\mathcal{P}}(I^{cml}, I^{(i)}),$$
that is: apply Protocol~\ref{protocol:single-fold}, using the current $I^{cml}$ as the first input and the new relaxed instance $I^{(i)}$ as the second input. During each step, we need a new $r = r^{(i)}$ for our ``folding randomness.'' 
\item $\mathcal{V}$ builds the committed relaxed instance
$\overline{I^{(i)}} = (\overline{T^{(i)}}, u^{(i)},  \overline{E^{(i)}}, R^{(i)} = R)$
out of commitments sent from $\mathcal{P}$. Using $\overline{I^{(i)}}$, $\mathcal{V}$ overwrites
$$ \overline{I^{cml}} \leftarrow \texttt{SingleFold}_{\mathcal{V}}(\overline{I^{cml}}, \overline{I^{(i)}}).$$
    \end{itemize}
    \item  After $N$ steps, the folding is complete. 
\end{enumerate}
    \end{protocol}
\end{tcolorbox}

We conclude with two remarks:
\begin{itemize}
    \item Note that $R$ does not need to be committed; it is, after all, sent by $\mathcal{V}$.
    \item This protocol can easily be generalized to the framework where we have several rounds of verifier input. A single step (2) of Protocol~\ref{protocol:full-cgvi} then involves building up from $T_0^{(i)}$ to $T = T_k^{(i)}$ after $k$ steps, each step making the partial commitments to $T_i$ right before asking for a new verifier input $S_i$. In our work, we do not use this full generality because of the notational burden and the fact that lookups only require a single round of verifier input.
\end{itemize}

\section{\textbf{Origami}: Folding Lookups}
\label{sec:lookups}

We now come to the main application. In this section, we describe an explicit folding scheme for a Halo2 lookup argument as a special case of our setup of custom gates with verifier randomness. 

\subsection{Lookups}
\label{sec:lookups-prelim}

Given two vectors of numbers (actually, elements of some large finite field $\mathbb{F}$) $A = (a_1, \ldots, a_m)$ and $S = (s_1, \ldots, s_n)$, a \emph{lookup argument} is an argument that each\footnote{$A$ does not have to contain every element of $S$. Also, $A$ can contain duplicates (as can $S$). But every element of $A$ has to appear somewhere in $S$.} $a_i$ is equal to some $s_j$. One common use case is evaluating a function by lookup table.  In this case, the vector $S$ encodes the values of some function, while $A$ encodes some claimed evaluations of the function at specific points. 

There have been several different implementations of lookup arguments, starting with Plookup \cite{plookup}. The lookup we are trying to fold is Halo2 lookups \cite[4.1]{halo2}, which we summarize as follows:
\begin{enumerate}
    \item  $\mathcal{P}$ finds a permutation $A'^{}$ of $A^{}$ and a permutation $S'^{}$ of $S^{}$ such that, for each $j$, either $a'^{}_j = a'^{}_{j-1}$ or $a'^{}_j = s'^{}_j$.
    \item $\mathcal{P}$ and $\mathcal{V}$ carry out a \emph{grand product protocol} to prove that $A'^{}$ is a permutation of $A^{}$ and $S'^{}$ is a permutation of $S^{}$.  
    \begin{enumerate}
        \item First, $\mathcal{V}$ sends $\mathcal{P}$ random challenges $\beta^{}$ and $\gamma^{}$.
        \item Based on those random challenges, $\mathcal{P}$ creates two new vectors $W^{}$ and $Z^{}$, the ``grand product vectors''.
    \end{enumerate}
\item In order to verify the lookup, $\mathcal{V}$ checks a list $P$ of polynomial identities involving the vectors $\{A^{}, A'^{}, S^{}, S', W^{}, Z^{}$\}, the random challenge scalars $\beta^{}$ and $\gamma^{}$, and constant and public $Q^*$ vectors (used for zero-knowledge): \begin{enumerate}
    \item $(1 - Q^{blind} - Q^{last}) \cdot \left (  Z[-1] (A' + \beta) - Z (A + \beta) \right ) = 0$
\item $(1 - Q^{blind} - Q^{last}) \cdot \left (  W[-1] (S' + \gamma) - W (S + \gamma) \right ) = 0$
\item $Q^{last} \cdot (Z^2 - Z) = 0$
\item $Q^{last} \cdot (W^2 - W) = 0$
\item $(1 - Q^{blind} - Q^{last}) (A' - S') (A' - A'[1]) = 0$
\item $Q^0 \cdot (A' - S') = 0$
\item $Q^0 \cdot (Z - 1) = 0$
\item $Q^0 \cdot (W - 1) = 0$.    
\end{enumerate}
\end{enumerate}
We call these equations $f_1 = 0$ through $f_8 = 0$ respectively. \begin{enumerate}
    \item Recall that each polynomial equation $f_i = 0$ corresponds to $n$ polynomial equations of the same form 
$$f_{i, j} (a_1, \ldots, a_n, a'_1, \ldots, a'_n, \ldots, w_1, \ldots, w_n, \beta, \gamma) = 0,$$
where each capital letter $X \in \{A, A', S, S', Z, W\}$ corresponds to substituting $x_j$ for $f_{i, j}$, $x_j$ being the $j$-th coordinate for the vector $X$. The coordinates of the $Q^*$ vectors, the $q^*_i$'s, do not appear as arguments since they are constants. As an example, the third equation $f_3 = Z^2 - Z = 0$ is shorthand for $n$ constraints of the form $f_{3, i}(\cdots) = z_i^2 - z_i = 0$ for each $i$, where $z_i$ is the $i$-th coordinate of $Z$. 
\item Recall that e.g.\ $Z[-1]$ means that when we unpack equation $f_{i, j}$, instead of substituting $z_{j}$, we substitute $z_{j - 1}$. As an example, $f_1 = 0 $ corresponds to $n$ constraints of the form 
$$f_{1, j} (\cdots) = (1 - q^{blind}_j - q^{last}_j)\left (  z_{j-1} (a'_j + \beta) - z_j (a_j + \beta) \right ) = 0.$$
\end{enumerate} 

The main point is that this setup fits within our framework of custom gates with verifier input. The structure $P$ consists of the $8$ polynomial identities $f_i$ mentioned in the last step, and the verifier input are the $\beta$ and $\gamma$ that $\mathcal{P}$ needs to generate $W$ and $Z$.

\begin{remark} [Zero-knowledge] For the scope of our presentation, we will not be concerned about zero-knowledge and the details of the $Q^*$ vectors. Readers interested in zero-knowledge should refer to   \cite[4.1.1]{halo2}. As a short summary:
\begin{itemize}
    \item Let $t=2$.  (In general, $t$ is the maximum number of distinct ``shifts'' that occur in any of the polynomial constraints.  Since we have $Z$ and $Z[-1]$ in the same constraint -- meaning the polynomial $f_{i, j}$ involves both $Z_j$ and $Z_{j-1}$ -- we take $t=2$ for this protocol.)
    \item To implement zero knowledge, only the first $n-t$ rows of the vectors $A, S, \ldots$ will participate in the lookup.  The last $t$ rows will be chosen by the prover $\mathcal{P}$ at random.
    \item $Q^0$: $q^0_0 = 1$ and $q^0_i = 1$ for all $i \neq 0$.  In \cite{halo2}, this $Q^0$ is denoted $\ell_0$ and called a \textit{Lagrange basis polynomial}.
    \item $Q^{blind}$: $q^{blind}_i = 0$ for $n-t \leq i \leq n-1$, and $1$ otherwise (i.e.\ $Q^{blind}$ selects the last $t$ rows).
\item  $Q^{last}$: $q^{last}_i = 1$ for $i = n-t-1$, and $0$ otherwise (i.e.\ $Q^{last}$ selects the $(n-t-1)$-st row, the last row before the ``blind'' rows).
\end{itemize}
\end{remark}

\subsection{Sketch and Example}

Before giving the specifics, we give an example of our procedure. Suppose we want to verify that $3, 7, 3, 5$ are odd numbers, while $6, 4, 4, 4$ are even. We may encode these lookups as $A^{(i)}$ and $S^{(i)}$, as in Table~\ref{tab:folding-example}.

\begin{table}[]
    \centering
    \begin{tabular}{|c|c||c|c|}
    \hline
 $A^{(1)}$ & $S^{(1)}$ & $A^{(2)}$ & $S^{(2)}$\\
\hline
3 & 1 & 6 & 2\\
\hline
7 & 3 & 4 & 4\\
\hline
3 & 5 & 4 & 6\\
\hline
5 & 7 & 4 & 8\\
\hline
    \end{tabular}
    \caption{Two lookups that we may want to fold together.}
    \label{tab:folding-example}.
\end{table}

First of all, the prover computes permutations $A'^{(1)}, S'^{(1)}, A'^{(2)}, S'^{(2)}$ such that each entry of $A'^{(1)}$ is equal, either to the previous entry of $A'^{(1)}$, or to the same entry of $S'^{(1)}$.  Similarly for $A'^{(2)}$ and $S'^{(2)}$. See Table~\ref{tab:folding-example-prime} for one particular way to do this for our example. We can check that:
\begin{itemize}
    \item $A'^{(1)} = (3, 3, 5, 7)$ is a permutation of $A^{(1)} = (3, 7, 3, 5)$, and similarly for the other three columns.
    \item Each entry of $A'^{(1)}$ is either the same as the entry of $S'^{(1)}$ next to it, or a repeat of the previous entry of $A'^{(1)}$.  The first entry $3$, the $5$ and the $7$ are all the same as the adjacent entries of $S'^{(1)}$, while the second $3$ is a repeat of the first.
    \item The same holds for $A'^{(2)}$.
\end{itemize}

\begin{table}[]
    \centering
    \begin{tabular}{|c|c||c|c|}
    \hline
 $A'^{(1)}$ & $S'^{(1)}$ & $A'^{(2)}$ & $S'^{(2)}$ \\ 
\hline
3 & 3 & 4 & 4\\
\hline
3 & 1 & 4 & 8\\
\hline
5 & 5 & 4 & 2\\
\hline
7 & 7 & 6 & 6\\
\hline
    \end{tabular}
    \caption{Example of creating the permuted versions of $A$ and $S$.}
    \label{tab:folding-example-prime}.
\end{table}

The combination of $A, S, A', S'$ (for either $i$) would be considered as our $T_0$ from Section~\ref{sec:cgvi}, the partial witness computable by $\mathcal{P}$ without further information. At this point, as in Protocol~\ref{protocol:full-cgvi}, the prover sends the verifier commitments to all eight vectors above. 

Now, $\mathcal{V}$ replies with random challenges $R^{(1)} = (\beta^{(1)}, \gamma^{(1)})$ and $R^{(2)} = (\beta^{(2)}, \gamma^{(2)}).$ $\mathcal{P}$ uses them to create grand products $Z^{(1)}, W^{(1)}, Z^{(2)}, W^{(2)}$. We defer the details of the algebra to later in this section; what's important to know is just that the lookup problem can be translated into checking some polynomial conditions $$f_i(\beta^{(1)}, \gamma^{(1)}, A^{(1)}, S^{(1)}, A'^{(1)}, S'^{(1)}, W^{(1)}, Z^{(1)}) = 0$$ and $$f_i(\beta^{(2)}, \gamma^{(2)}, A^{(2)}, S^{(2)}, A'^{(2)}, S'^{(2)}, W^{(2)}, Z^{(2)}) = 0.$$

So far we have two lookup arguments: one for the odd numbers, and one for the evens.  Now let's fold them together, as in Protocol~\ref{protocol:single-fold}. First, the prover will send the verifier commitments to some cross-terms $B_1, B_2, B_3, B_4$. 
Then the verifier sends the prover a random challenge $r$.  (In our example, let's imagine $r = 100$.) Finally, the prover computes the linear combinations $A^{(1)} + r A^{(2)}, S^{(1)} + r S^{(2)}$ and so forth.

\begin{table}[]
    \centering
    \begin{tabular}{|c|c|c|c|}
    \hline
$A_1 + r A^{(2)}$ & $S^{(1)} + r S^{(2)}$ & $A'^{(1)} + r A'^{(2)}$ & $S'^{(1)} + r S'^{(2)}$ \\ 
\hline
603 & 201 & 403 &403\\
\hline
407 & 403 & 403 & 801\\
\hline
403 & 605 & 405 & 205\\
\hline
405 & 807 & 607 & 607\\
\hline
    \end{tabular}
    \caption{The folded trace / witness of the two lookups. Notice that the result is no longer a lookup.}
    \label{tab:folding-example-folded}.
\end{table}

To see the folded example, see Table~\ref{tab:folding-example-folded}. Something strange happened here: The ``lookup'' and ``permutation'' relations are no longer satisfied.  The number $603$ appears in $A^{(1)} + r A^{(2)}$ but not in $S^{(1)} + r S^{(2)}$.  The vector $A'^{(1)} + r A'^{(2)}$ is not a permutation of $A^{(1)} + r A^{(2)}$.  It looks like we've thrown out all the nice properties that made the lookup argument work.

But in fact polynomials save the day!  The ``folded'' vectors will still satisfy a polynomial identity that we'll be able to write down, and that's what the prover will verify.

\subsection{The Procedure}

We now show how our protocol works in the case of lookups. First, we relax $f_i$ into the following homoegenous equations with $u$ and $E_i$:
\begin{enumerate}
\item $(1 - Q^{blind} - Q^{last}) \cdot \left (  Z([-1]) (A' + \beta) - Z (A + \beta) \right ) = E_1$
\item $(1 - Q^{blind} - Q^{last}) \cdot \left (  W([-1]) (S' + \gamma) - W (S + \gamma) \right ) = E_2$
\item $Q^{last} \cdot (Z^2 - u Z) = E_3$
\item $Q^{last} \cdot (W^2 - u W) = E_4$,
\item $(1 - Q^{blind} - Q^{last}) (A' - S') (A' - A'[1]) = E_5$,
\item $Q^0 \cdot (A' - S') = 0$
\item $Q^0 \cdot (Z - u) = 0$
\item $Q^0 \cdot (W - u) = 0$.
\end{enumerate}
For example, recall that $f_3$ used to encode $n$ constraints of form 
$$f_{3, i} = q_i^{last} (z_i^2 - z_i) = 0;$$
the new relaxed $f_3$ would instead encode $n$ constraints of the form
$$f_{3, i} = q_i^{last} (z_i^2 - z_iu) = e_{3, i}.$$

Note that each of the first $5$ quadratic equations give a slack term, and the linear equations do not. Applying Proposition~\ref{prop:delta} gives the corresponding cross-terms
\begin{eqnarray*} 
B_{1} & = & (1 - Q^{blind} - Q^{last}) \cdot \\
& & \big( (Z^{(1)}[-1] (A'^{(2)} + \beta^{(2)}) + Z^{(2)}[-1] (A'^{(1)} + \beta^{(1)}) \\
& & - Z^{(1)} (A^{(2)} + \beta^{(2)}) - Z^{(2)} (A^{(1)} + \beta^{(1)}) \big)  \\
B_{2} & = & (1 - Q^{blind} - Q^{last}) \cdot \\
& & \big( W^{(1)}[-1] (S'^{(2)} + \beta^{(2)}) + W^{(2)}[-1] (S'^{(1)} + \beta^{(1)}) \\
& & - W^{(1)} (S^{(2)} + \beta^{(2)}) - W^{(2)} (S^{(1)} + \beta^{(1)}) \big) \\
B_{3} & = & Q^{last} \cdot (2 Z^{(1)} Z^{(2)} - u^{(1)} Z^{(2)} - u^{(2)} Z^{(1)} \\
B_{4} & = & Q^{last} \cdot (2 W^{(1)} W^{(2)} - u^{(1)} W^{(2)} - u^{(2)} W^{(1)} \\
B_5 & = & (1 - Q^{blind} - Q^{last}) \cdot \\
& & (2 A'^{(1)} A'{(2)} - S'^{(1)} A'{(2)} - S'^{(2)} A'^{(1)} - A'^{(1)} A'{(2)}[1] - A'^{(2)} A'^{(1)}[1] + \\
& & S'^{(1)} A'^{(2)}[1] + S'^{(2)} A'{(1)}[1]).
\end{eqnarray*}
For example, $B_3$ corresponds to, as $j$ runs from $1$ to $n$, 
$$ B_{3, j}(a_{1, 1}, \ldots, z_{n, 1}, a_{1, 2}, \ldots, z_{n, 2}, u_1, u_2, \beta_1, \beta_2, \gamma_1, \gamma_2) = q^{last}_j (2z_{j, 1}z_{j, 2} - u_1 z_{j, 2} - u_2 z_{j, 1}).$$

Recall that the main idea here is that we really want the $f_i$ equations to satisfy constraints of the form
$$f_i(X_1 + rX_2) = f_i(X_1) + r f_i(X_2),$$
where $X_1$ and $X_2$ are shorthand meaning ``all the arguments,'' but for the quadratic $f_i$ we have to settle for 
$$ f_i(X_1 + rX_2) = f_i(X_1) + r^2 f_i(X_2) + r B_{i}(X_1, X_2).$$
To see this explicitly with $f_3 = Z^2 - uZ,$ we compute
\begin{align}
f_{3}(X_1 + rX_2) & = (Z_1 + rZ_2)^2  - (u_1 + ru_2) (Z_1 + rZ_2) \\
& = (Z_1^2 + 2r Z_1 Z_2 + r^2Z_2^2) - u_1Z_1 - ru_1Z_2 - ru_2Z_1 - r^2 u_2 Z_2 \\
& = (Z_1^2 - u_1Z_1) + r^2(Z_2^2 - u_2Z_2) + r(2Z_1Z_2 - u_1Z_2 - u_2Z_1) \\
& = f_{3}(X_1) + r^2 f_{3}(X_2) + r B_{3}(X_1, X_2),
\end{align}
as desired.


We can now simply follow Protocol~\ref{protocol:full-cgvi}, with $T^0, R, T$ defined as follows:
\begin{enumerate}
    \item First, $\mathcal{P}$ is given the instance $A$ and $S$. $\mathcal{P}$ then computes $A'$ and $S'$ without further help, so the $n \times 4$ partial trace $T_0$ would consist of the information in $(A, S, A', S')$, which can be done via e.g.\ concatenation.
    \item After committing to $T_0$ to obtain $\overline{T_0}$, $\mathcal{P}$ receives the verifier randomness $R = (\beta, \gamma)$.
    \item $\mathcal{P}$ is now able to compute the $n \times 6$ full trace $T$ by making $2$ more columns $Z$ and $W$.
\end{enumerate}

\section{Comparison with Other Protocols}

\subsection{Nova}

Nova \cite{nova} is a folding scheme for R1CS circuits.  The R1CS structure is a special case of an AIR, where there is a single constraint polynomial $f_1$, of degree 2, having a specific form.  Because the polynomial is of degree $d=2$, there is only a single cross-term $B = B_{1, 1}$.  The folding scheme described here generalizes Nova: if you apply this scheme to R1CS systems, you recover the original Nova folding scheme.

The idea that Nova-style folding can be generalized to arbitrary custom gates was introduced in \cite[Section 3.3]{sangria}.

\subsection{Sangria}

Sangria \cite{sangria} had outlined an approach to custom gates. Our lookup protocol is almost, but not quite, a special case of that procedure, mainly because of the roles of $\beta$ and $\gamma$. Instead, what we have here is a special case of ``custom gates with verifier randomness,'' a concept that's a slight generalization of Sangria's approach to custom gates. 

The argument for \emph{knowledge soundness} (a cheating prover cannot convince the verifier to accept a folded proof unless the prover actually knows $N$ satisfying witnesses) is similar to Sangria. We give a proof of knowledge soundness in Appendix~\ref{app:soundness}.

In outline, the proof (very similar to as done in Sangria and Nova) is as follows.  The idea is to imagine an \emph{extractor} that interacts with the prover.  The extractor is allowed to rewind the prover to a previous state.  In practice, this means the extractor (playing the role of verifier) can send the prover different challenges $r$, and see how the prover responds.  Like in Nova and Sangria, by testing enough different values $r$ and doing a bit of algebra, the extractor can recover the witnesses $N$ that were folded together. Once the extractor can recover the $N$ folded witnesses, since Halo2 lookups themselves are knowledge sound, we know the prover must know $N$ valid lookup witnesses.

Our proof of knowledge soundness in the general scheme is only a bit more complicated than that of Sangria, partly due to the polynomials of arbitrary degree (although again Sangria had an outline already in Section 3.3) and partly due to the fact that we have ``verifier randomness,'' as stated above. We have to take care that the verifier-provided randomness $\beta$ and $\gamma$ does not mess things up.  

\subsection{Related Work: Moon Moon and Protostar}

In work in progress (\cite{moonmoon} and personal communication), Lev Soukhanov has proposed Moon Moon, a more powerful approach to combine folding with verifier input.
In Moon Moon, suppose there are $N$ instance-witness pairs to be folded.
The prover first commits to all $N$ partial witnesses, 
then receives (once) a piece of verifier randomness $R$.
The prover then uses $R$ to compute the remaining columns of the full witness.

Moon Moon allows, for example, an extended permutation argument across $N$ execution traces.
The prover commits to all entries of both permutations,
spread across $N$ partial witnesses, 
then computes a single Fiat--Shamir hash
and uses it to construct a grand product polynomial, which proves the permutation constraint.

Protostar \cite{protostar} is a proving system that uses a folding scheme for higher-degree polynomial gates.  
Additionally, Protostar is a \emph{non-uniform} IVC:
it can prove a sequence of computations of the form 
\[ z_{i+1} = F_{k_i} (z_i), \]
where $F_1, \ldots, F_n$ are $n$ different computations encoded by $n$ different circuits.
This has applications to proving the output of a virtual machine which allows $n$ different operations.

\section{Acknowledgments}
We thank Nicolas Mohnblatt, Lev Soukhanov, Yi Sun, and Jonathan Wang for valuable discussions. 

\bibliographystyle{plain}
\bibliography{bibliography}

\appendix

\section{Knowledge Soundness}
\label{app:soundness}

In this section we will prove knowledge soundness of a noninteractive version of Protocol \ref{protocol:full-cgvi}, constructed by the Fiat--Shamir transformation.  We will work in the algebraic group model, and assume the commitment scheme used is KZG polynomial commitments.  We begin by recalling some preliminary concepts.

Intuitively, knowledge soundness is the claim that, if the prover convinces the verifier to accept a proof, then the prover ``knows'' a satisfying witness.

The claim is generally formalized as follows: a protocol is \emph{knowledge sound} if, for every (possibly dishonest) prover $\mathcal{P}$ and honest verifier $\mathcal{V}$, there exists an extractor $\mathcal{E}$, 
such that whenever $\mathcal{P}$ convinces $\mathcal{V}$ to accept a proof,
$\mathcal{E}$ can determine a satisfying witness (with all but negligible probability) by interacting with $\mathcal{P}$.
To make this precise, we need to say exactly what sorts of ``interaction'' are allowed.

We will work in the \emph{algebraic group model}.  That is, whenever the prover outputs an element $g$ of a group $G$ (e.g.\ an element of $\mathbb{F}$), the prover also outputs (but does not send to the verifier) a representation
\[ g = \sum_{i=1}^n a_i G_i \]
of $g$ in terms of previously-seen elements $G_i \in G$.  We write this representation as $[g]$.
The representation is not available to the verifier, but it is available to the extractor.

We will assume that the commitment scheme used is KZG polynomial commitments \cite{kzg}.

\begin{tcolorbox}
    We provide a brief summary of the KZG polynomial commitment scheme.

    The KZG commitment scheme allows a prover to commit to polynomials of degree up to some fixed $t \geq 0$.
    The scheme relies on an \emph{trusted setup}, which outputs $(g, {\alpha}g, {\alpha^2}g, \ldots, {\alpha^t}g)$, for some fixed
    $g \in G$ and $\alpha \in \mathbb{Z}$; the integer $\alpha$ itself is kept secret.

    To commit to a polynomial
    \[ p(T) = \sum_{i = 0}^{t} a_i T^i, \]
    the prover sends
    \[ \overline{p(T)} = \sum_{i=0}^t a_i ({\alpha^i} g). \]
    
\end{tcolorbox}

This commitment scheme has the property that, in the algebraic group model, 
if the prover sends a commitment to some $p(T)$, then the extractor has access to $p(T)$ (except with negligible probability).
Indeed, if the prover has honestly committed some $p(T)$, then the algebraic representation of the commitment in terms of the previously-seen group elements $\alpha^i g$ will take the form
\[ \overline{p(T)} = \sum_{i=0}^t a_i (\alpha^i g), \]
and the coefficients $a_i$ can be directly read off from this.  
For future reference, we let $[\mathrm{KZG}(p(T))]$ denote the algebraic representation
\[ \mathrm{KZG}(p(T)) = \sum_{i=0}^t a_i (\alpha^i g) \]
of the KZG commitment to the polynomial $p(T)$.

Protocol \ref{protocol:full-cgvi} is a public-coin interactive protocol between a prover $\mathcal{P}$ and a verifier $\mathcal{V}$.  The parties $\mathcal{P}$ and $\mathcal{V}$ send messages to each other in turn, in a total of $2r+1$ rounds, for some $r$: $\mathcal{P}$ sends some $a_1$, then $\mathcal{V}$ sends a challenge $c_1$, then $\mathcal{P}$ sends a message $a_2$, and so forth.  The verifier's messages are randomly chosen \emph{challenges}.
In this setting, the Fiat--Shamir transform replaces $\mathcal{V}$ with a cryptographic hash function:
whenever Protocol \ref{protocol:full-cgvi} calls for a random challenge from $\mathcal{V}$, 
the Fiat--Shamir prover $\mathcal{P}$ computes a cryptographic hash of the transcript of all values that have been sent so far, and uses that hash as the random challenge.

Using Fiat--Shamir gives a dishonest prover a new avenue of attack: 
the prover can rewind the verifier to get new challenges.
This makes proof of knowledge soundness more difficult.
We will use ideas introduced in \cite{gt}
to bound the probability that a dishonest prover
can produce an accepted transcript.

\cite[Theorem 2]{gt} relates knowledge soundness of a Fiat--Shamir protocol to \emph{state-restoration witness extended emulation} (sr-wee) soundness of the corresponding interactive protocol.  
For a formal definition of sr-wee soundness, see \cite[\S 4]{gt}.
Informally, the definition of sr-wee soundness is as follows:
for every prover $\mathcal{P}$, there exists an \emph{emulator} $\mathcal{E}$,
such that if $\mathcal{P}$ can produce an accepted transcript by repeated interaction with $\mathcal{E}$, 
then $\mathcal{E}$ can produce a satisfying witness with all but negligible probability.
Here $\mathcal{E}$, like the extractor above, 
interacts with $\mathcal{P}$ in the role of the verifier,
but is also given access to the algebraic representation $[g]$
of every group element sent by $\mathcal{P}$.
Meanwhile, $\mathcal{P}$ is allowed to rewind the emulator, 
sending different group elements to try to get a favorable challenge back from $\mathcal{E}$.
By \cite[Theorem 2]{gt}, if an interactive protocol satisfies sr-wee soundness, 
then the corresponding noninteractive protocol obtained by Fiat--Shamir is itself sound.

\cite[Theorem 1]{gt} provides a framework for proving sr-wee soundness.
Suppose we have the following:
\begin{itemize}
    \item For every partial transcript $\tau = ([a_1], c_1, \ldots, [a_i])$, a small set of \emph{bad challenges} $\mathrm{BadCh}(\tau)$, and
    \item An \emph{extractor function} $\mathrm{e}$ that takes as input an accepted transcript $([a_1], c_1, \ldots, [a_{r+1}])$ and outputs, with high probability, a valid witness $w$.
\end{itemize}
Then \cite[Theorem 1]{gt} shows that the protocol is sr-wee sound.  
Precisely, suppose that:
\begin{itemize}
    \item For every $\tau$, we have $\left | \mathrm{BadCh}(\tau) \right | \leq \epsilon \left | \mathrm{Ch} \right |$, where $\mathrm{Ch}$ is the set from which challenges are randomly drawn, and
    \item The probability in SRS that $\mathcal{P}$ can produce an accepted transcript $([a_1], c_1, \ldots, [a_{r+1}])$, with $c_i \not \in \mathrm{BadCh}([a_1], \ldots, [a_i])$ for every $i$, is at most $p_{fail}$.  (SRS \cite[\S3]{gt} means that $\mathcal{P}$ is allowed to rewind $\mathcal{V}$.)
\end{itemize}
Then the protocol has sr-wee soundness error of at most
\[ q\epsilon + p_{fail}, \]
where $\mathcal{P}$ is allowed to make at most $q$ queries to the (Fiat--Shamir) oracle.

\begin{thm}
    Suppose $\mathrm{IP}_0$ is an interactive protocol for some relation $\mathcal{R}$, of the following form.

    \begin{tcolorbox}
    The protocol $\mathrm{IP}_0$.
    
    \begin{itemize}
        \item $\mathcal{P}$ sends some $T_0$ to $\mathcal{V}$.
        \item $\mathcal{V}$ sends a random challenge $R$ to $\mathcal{P}$.
        \item $\mathcal{P}$ sends some $T_1$ to $\mathcal{V}$.
        \item Writing $T$ for the concatenation of $T_0$, $R$ and $T_1$, $\mathcal{V}$ verifies the polynomial constraints
        \[ f_j(T) = 0 \]
        for $j = 1, \ldots, \ell$.  If these constraints are satisfied, we say that $(T_0, R, T_1)$ is an \emph{accepted transcript} for $\mathrm{IP}_0$.
    \end{itemize}
    \end{tcolorbox}
Assume:
\begin{itemize}
    \item For each $T_0$, there is a set $\mathrm{BadCh}_0(T_0)$ of \emph{bad challenges}, such that $\left | \mathrm{BadCh}_0(T_0) \right | \leq \epsilon \left | \mathrm{Ch} \right |$, 
    \item The degree $d_j$ of each constraint polynomial $f_j$ satisfies $d_j \leq \epsilon \left | \mathrm{Ch} \right |$, and
    \item There is an extractor function $\mathrm{e}$ such that the probability in SRS that a cheating prover $\mathcal{P}$ can create an accepted transcript $(T_0, R, T_1)$, such that $R \not \in \mathrm{BadCh}_0(T_0)$ but $\mathrm{e}_0 ([T_0], R, [T_1])$ is not an accepted transcript, is at most $p_{fail}$.
\end{itemize}

Let $\mathrm{IP}$ be Protocol \ref{protocol:full-cgvi}, with the following two modifications.

In place of step 2 (a-c), $\mathcal{P}$ and $\mathcal{V}$ carry out $\mathrm{IP}_0$ (where $T$ is the concatenation $T_0 T_1$).

At the end of the protocol, $\mathcal{P}$ reveals $T^{cml}$ and $E^{cml}$.  Then $\mathcal{V}$ accepts the transcript if the KZG commitments to $T^{cml}$ and $E^{cml}$ agree with the committed values $\overline{T^{cml}}$ and $\overline{E^{cml}}$ already calculated, and the polynomial constraints
\[ f_j^{homog}(u^{cml}, T^{cml}) = E^{cml} \]
are all satisfied.

Then $\mathrm{IP}$ is sr-wee sound for the relation
\[ \mathcal{R}^{cml} = \left \{ (x_1, w_1, x_2, w_2, \ldots, x_N, w_N) | (x_i, w_i) \in \mathcal{R} \text{ for all $1 \leq i \leq N$} \right \}, \]
with soundness error at most $(2N+1) \epsilon + p_{fail} + \operatorname{Negl}(\lambda)$.
\end{thm}

\begin{proof}
We will use \cite[Theorem 1]{gt}.  We need to define functions $\mathrm{BadCh}_i$ (one for each round of verifier challenge) and $\mathrm{e}$.

We need to define bad challenges $\mathrm{BadCh}$ for two types of verifier challenge: challenges $R^{(i)}$ in step 2 of Protocol \ref{protocol:full-cgvi}, which come from protocol $\mathrm{IP}_0$,
and folding challenges $r^{(i)}$ in step 3, which come from Protocol \ref{protocol:single-fold}.

For challenges $R^{(i)}$, we simply let $\mathrm{BadCh}(\overline{T_0^{(i)}})$ be the set of bad challenges for $\mathrm{IP}_0$ given by hypothesis.

For folding challenges $r^{(i)}$, we distinguish two cases.

Recall that the transcript $\tau$ of all messages sent before $\mathcal{V}$'s challenge $r^{(i)}$ contains the two committed instances to be folded $I^{cml} = (u^{cml}, \overline{T}^{cml}, \overline{E}^{cml})$ and $I^{(i)} = (u^{(i)}, \overline{T}^{(i)}, \overline{E}^{(i)})$, as well as $\mathcal{P}$'s commitments $\overline{B_{i, k}}$ to cross terms.

Since we are working in the AGM, the function $\mathrm{BadCh}$ also has access to representations of all prover output in terms of previously-seen group elements; these representations are written as $[\overline{T}^{cml}], [\overline{E}^{cml})],$ and so forth.

We say that $\tau$ is \emph{good} if:
\begin{enumerate}
    \item The algebraic representations of all commitments sent by $\mathcal{P}$ take the form of KZG commitments. 
    In other words, there exist $T^{cml}, E^{cml}, T^{(i)}, E^{(i)}, B_{i, k}$ such that
    \[ [\overline{T}^{cml}] = [\mathrm{KZG}(T^{cml})], \]
    and so forth, and
    \item The two instance-witness pairs $(u^{cml}, T^{cml}, E^{cml})$ and $(u^{(i)}, T^{(i)}, E^{(i)})$ are satisfying instance-witness pairs, with cross-terms $B_{i, k}$.

    In other words, we require that
    \begin{equation}
    \label{identity_r}
    f_j^{homog}(T^{cml} + r T^{(i)}) = E^{cml} + r^{d_i} E^{(i)} + \sum_{k=1}^{d_j - 1} r^k B_{j, k}, 
    \end{equation}
    identically as polynomials in $r$, for each polynomial relation $f_j$.
\end{enumerate}

If $\tau$ is \emph{good}, then we take $\mathrm{BadCh}(\tau)$ to be the empty set.

If $\tau$ fails to satisfy condition (1) above (i.e.\ the prover has not sent KZG commitments), we also take $\mathrm{BadCh}(\tau)$ to be the set of $r$ for which the algebraic representations $[\overline{T^{cml}}] + r [\overline{T^{(i)}}], [\overline{E^{cml}}] + r [\overline{E^{(i)}}], [\overline{B_{i, k}}]$
all have the form of Kate commitments.
By linearity of Kate commitment, $\left | \mathrm{BadCh}(\tau) \right | \leq 1$.

Otherwise, $\tau$ satisfies condition (1) but not condition (2): the prover has sent commitments to invalid witnesses.
In this case, choose some $j$ for which Equation \ref{identity_r} does not hold identically.  The equality can only hold for at most $d_j$ values of $r$; let $\mathrm{BadCh}(\tau)$ be the set of $r$ for which Equation \ref{identity_r} holds.

Next, we define the extractor function $\mathrm{e}$.
At each step $i$, if the algebraic representation of the prover's output $\overline{T^{(i)}}$ has the form of a KZG commitment
\[ [\overline{T^{(i)}}] = [\mathrm{KZG}(T^{(i)})], \]
then $\mathrm{e}$ applies $\mathrm{e}_0$ to the committed value $T^{(i)}$, returning a purported witness $w^{(i)} = \mathrm{e}_0(T^{(i)})$.
If not, $\mathrm{e}$ simply returns a null value $w^{(i)} = \emptyset$.
The final output of $\mathrm{e}$ is the tuple $(w^{(1)}, \ldots, w^{(n)})$.

Now we need show that the probability that $\mathcal{P}$ can produce some accepted transcript $\tau$, containing no bad challenges, and such that $\mathrm{e}(\tau)$ is not a valid witness, is negligible in the security parameter $\lambda$.

Since $\mathcal{V}$ accepts $\tau$, 
the final $\overline{(T^{cml})}$ is of the form $\overline{(T^{cml})} = \mathrm{KZG}(T^{cml})$ of a KZG commitment, and similarly for $\overline{(E^{cml})}$.
On the other hand, $\overline{(T^{cml})}$ comes with an algebraic representation $[\overline{(T^{cml})}]$ as a linear combination of the representations $[\overline{T^{(i)}}]$ in the AGM.
Since we have assumed the discrete logarithm problem is hard in $G$,
we have $[\overline{(T^{cml})}] = [\mathrm{KZG}(T^{cml})]$
except with probability negligible in the security parameter $\lambda$ 
(because if this equality did not hold, then $\mathcal{P}$ would have discovered
a nontrivial linear relation among elements of $G$).

So now suppose the algebraic representation $[\overline{(T^{cml})}]$
has the form of a KZG commitment.
By reverse induction we see that, since there are no bad challenges, 
all the folded values $[\overline{(T^{(i)})}]$ and $[\overline{(E^{(i)}}]$ are KZG commitments as well,
of the form $[\mathrm{KZG}(T^{(i)})]$ and $[\mathrm{KZG}(E^{(i)})]$, respectively.
Furthermore, at each step, $T^{cml}$ is replaced with $T^{cml} + r^{(i)} T^{(i)}$, and similarly for $E^{cml}$.

Again by reverse induction, since $\tau$ contains no bad challenges, we see that
\[ f_j^{homog}(u^{cml}, T^{cml}) = E^{cml} \]
at each step, and
\[ f_j^{homog}(u^{(i)}, T^{(i)}) = E^{(i)} \]
for every $i$.

Thus, the transcript $T^{(i)}$ is accepted for the protocol $\mathrm{IP}_0$.
By hypothesis, the probability that $\mathrm{e}_0(T^{(i)})$ is not a satisfying witness is negligible in $\lambda$.

\end{proof}

\end{document}